\documentclass[a4paper,11pt]{article}
\usepackage{amsmath,amssymb,amsthm,tabularx,hyperref,authblk}
\usepackage[end]{algpseudocode}

\theoremstyle{plain}
\newtheorem{theorem}{Theorem}
\newtheorem{lemma}[theorem]{Lemma}
\newtheorem{corollary}[theorem]{Corollary}
\theoremstyle{definition}
\newtheorem{definition}[theorem]{Definition}

\newtheorem{property}[theorem]{Property}
\theoremstyle{remark}

\long\def\ignore#1{\vskip 0pt}

\def\Oh{\mathcal{O}}

\newcommand{\xx}{\$}
\newcommand{\A}{\Sigma}

\newcommand{\lcp}{\mathsf{lcp}}
\newcommand{\lcpzo}{\mathsf{lcp}_{01}}
\newcommand{\lcpz}{\mathsf{lcp}_{0}}
\newcommand{\lcpo}{\mathsf{lcp}_{1}}

\newcommand{\sort}{\mathsf{sort}}

\newcommand{\sa}{\mathsf{sa}}
\newcommand{\sazo}{\mathsf{sa}_{01}}
\newcommand{\SA}{\mathsf{sa}}
\newcommand{\LCP}{\mathsf{LCP}}

\newcommand{\bwt}{\mathsf{bwt}}
\newcommand{\mbwt}{\mathsf{bwt}}
\newcommand{\bwtzo}{\mathsf{bwt}_{01}}

\newcommand{\onex}{\mathbf{1}}
\newcommand{\zerox}{\mathbf{0}}

\newcommand{\tx}{\mathsf{t}}

\def\stri#1{\mbox{\sf #1}}

\newcommand{\tz}{\mathsf{t}_0}
\newcommand{\tone}{\mathsf{t}_1}
\newcommand{\tzo}{\mathsf{t}_{01}}
\newcommand{\txx}[1]{\mathsf{t}_{#1}}
\newcommand{\sxx}{\mathsf{s}}
\newcommand{\saxx}{\mathsf{sa}}
\newcommand{\nz}{{n_0}}
\newcommand{\none}{{n_1}}
\newcommand{\eosz}{\$_0}
\newcommand{\eosone}{\$_1}
\newcommand{\eosx}[1]{\$_{#1}}
\newcommand{\bwtz}{\mathsf{bwt}_0}
\newcommand{\bwto}{\mathsf{bwt}_1}
\newcommand{\bwtx}[1]{\mathsf{bwt}_{#1}}
\newcommand{\oneb}{{\bf 1}}
\newcommand{\zerob}{{\bf 0}}
\newcommand{\Bid}{\mathsf{Block\_id}}
\newcommand{\bid}{\mathsf{id}}
\newcommand{\avelcp}{\mathsf{avelcp}_{01}}

\newcommand{\bv}[1]{Z^{(#1)}}
\newcommand{\kh}{{b(h)}}
\newcommand{\sbot}{0}

\newcommand{\hm}{{\sf H\&M}}
\newcommand{\gap}{{\sf Gap}}
\newcommand{\useless}{irrelevant}

\begin{document}

\title{\bf From H\&M to Gap\\ for Lightweight {BWT} Merging}

\ignore{\author{Giovanni Manzini}

\institute{
    Computer Science Institute\\
    University of Eastern Piedmont, Alessandria, Italy\\
    \email{giovanni.manzini@uniupo.it}\\[1ex]
}}

\author[1,2]{Giovanni Manzini}
\affil[1]{Computer Science Institute, University of Eastern Piedmont, Italy}
\affil[2]{IIT-CNR, Pisa, Italy}

\date{}

\maketitle \thispagestyle{empty}

\begin{abstract}

Recently, Holt and McMillan [Bionformatics 2014, ACM-BCB 2014] have proposed
a simple and elegant algorithm to merge the Burrows-Wheeler transforms of a
family of strings. In this paper we show that the \hm\ algorithm can be
improved so that, in addition to merging the BWTs, it can also merge the
Longest Common Prefix (LCP) arrays. The new algorithm, called \gap\ because
of how it operates, has the same asymptotic cost as the \hm\ algorithm and
requires additional space only for storing the LCP values.

\end{abstract}


\section{Introduction} \label{sec:intro}

Compressed indices~\cite{NM-survey07} are core components of many data
intensive tools, especially in bioinformatics~\cite{LiH10,Reinert2015}. A
fundamental component of many compressed indices is the {Burrows Wheeler
transform} (BWT) of the data to be indexed, which is often complemented by
the {Longest Common Prefix} (LCP) array and (a sampling of) the {Suffix
Array}. Because of the sheer size of the data involved, the construction of
compressed indices is a challenging problem in itself. Although the final
outcome is a {\em compressed} index, construction algorithms can be memory
intensive and the necessity of developing {\em lightweight}, ie space
economical, algorithms was recognized since the very beginning of the
field~\cite{BurKar03,lcp_swat,MF02}. An alternative to space economical
algorithms are external memory construction algorithms, where the challenge
is to efficiently use the abundant external memory, typically by accessing
data in large blocks (see \cite{icabd/KarkkainenK14,jea/KarkkainenK16} and
references therein).

Many construction algorithms for compressed indices are designed for the case
the input consists in a single sequence; yet in many applications the data to
be indexed consist in a collection of distinct items (documents, web pages,
chromosomes, proteins, etc.). One can concatenate such items using distinct
end-of-file separators and index the resulting single sequence. However, this
is possible only for small collections and from the algorithmic point of view
it makes no sense to ``forget'' that the input consists of different items:
this is an additional information that algorithms should exploit to run
faster.

A case of great practical interest is the one where the input is a collection
of Next Generation Sequencing (NGS) reads, which typically consists in
millions of sequences of lengths ranging from a few hundreds to a few
thousands symbols. In this case the use of explicit distinct separators is
not feasible: A few algorithms have been therefore developed specifically for
this problem. The first one is the {\sf BCR} algorithm~\cite{tcs/BauerCR13}.
For a collection of $m$ strings of total length $n$ and maximum length $K$,
{\sf BCR} uses $\Oh(m\log(mK))$ bits of space and takes $\Oh(K\sort(m))$ time
where  $\sort(m)$ is the time to sort $m$ integers. {\sf BCR} can be modified
into an external memory algorithm that uses a negligible amount of RAM and a
I/O volume of $\Oh(mK^2)$ bits. Recently, {\sf BCR} has been extended to
compute the LCP arrays along with the BWTs in $\Oh(K(n + \sort(m))$
time~\cite{jda/CoxGRS16}.

Two other algorithms designed for NGS reads are {\sf CX1}~\cite{LiuLL14} and
{\sf ropeBWT}~\cite{bioinformatics/Li14a}. The former is designed to exploit
the computing power of modern GPUs, while the latter uses a dynamic data
structure (a {\sf B+} tree) to maintain partial BWTs so that its complexity
is $\Oh(n\log n)$ time. A comparison of the performance of these algorithms
for different sets of NGS reads are reported
in~\cite[Table~1]{bioinformatics/Li14a}.

Recently, Holt and McMillan~\cite{bcb/HoltM14,bioinformatics/HoltM14} have
presented a new approach for computing the BWT of a collections of sequences
based on the concept of merging: First the BWTs of the individual sequences
are computed (by any single-string BWT algorithm) and then they are merged,
possibly in multiple rounds as in the standard mergesort algorithm. The idea
of BWT-merging is not new~\cite{FGM10,Siren09} but Holt and McMillan's
merging algorithm is different, and much simpler, that the previous
approaches. For a constant size alphabet the algorithm in~\cite{bcb/HoltM14}
merges the BWTs of two sequences $\txx{0}$, $\txx{1}$ in $\Oh(n\cdot\avelcp)$
time where $n=|\txx{0}| + |\txx{1}|$ and $\avelcp$ is the average length of
the common prefix between suffixes of $\txx{0}$ and $\txx{1}$. The average
longest common prefix is $\Oh(n)$ in the worst case but $\Oh(\log n)$ for
random strings and for many real word datasets~\cite{LMS12}. The algorithm is
lightweight in the sense that it uses only $\Oh(n)$ bits in addition to the
space for its input and output.

In this paper we show that the \hm\ (Holt and McMillan) merging algorithm can
be modified so that, in addition to the BWTs, it can merges the LCP arrays as
well. The new algorithm, called \gap\ because of how it operates, has the
same asymptotic cost as \hm\ and uses additional space only for storing its
additional input and output, ie the LCP values.

\ignore{Our algorithm is loosely based on a technique first introduced
in~\cite{BellerGOS13} to compute the LCP values from the BWT of a single
string. Here, this technique is applied to simultaneously compute BWT and LCP
values.}

\section{Notation}\label{sec:notation}

Let $\txx{}[1,n]$ denote a string of length $n$ over a finite alphabet $\A$.
As is usual in the indexing literature we assume $\txx{}[n]$ is a symbol not
appearing elsewhere in $\txx{}$ and lexicographically smaller than any other
symbol. We write $\txx{}[i,j]$ to denote the substring $\txx{}[i] \txx{}[i+1]
\cdots \txx{}[j]$. If $j\geq n$ we assume $\txx{}[i,j] = \txx{}[i,n]$. If
$i>j$ or $i > n$ then $\txx{}[i,j]$ is the empty string. Given two strings
$\txx{}$ and $\sxx$ we write $\txx{} \preceq \sxx$ ($\txx{} \prec \sxx$) to
denote that $\txx{}$ is lexicographically (strictly) smaller than $\sxx$. As
usual we assume that if $\txx{}$ is a prefix of $\sxx$ then $\txx{} \prec
\sxx$. We denote by  $\LCP(\txx{},\sxx{})$ the length of the longest common
prefix between $\txx{}$ and $\sxx$.

The {\em suffix array} $\saxx[1,n]$ associated to $\txx{}$ is the permutation
of $[1,n]$ giving the lexicographic order of $\txx{}$'s suffixes, that is,
for $i=1,\ldots,n-1$, $\txx{}[\saxx[i],n] \prec \txx{}[\saxx[i+1],n]$. The
{\em longest common prefix} array $\lcp[1,n+1]$ is defined for $i=2,\ldots,n$
by
\begin{equation}\label{eq:lcpdef}
\lcp[i]=\LCP(\txx{}[\saxx[i-1],n],\txx{}[\saxx[i],n]),
\end{equation}
that is, the $\lcp$ array stores the length of the longest common prefix
between lexicographically consecutive suffixes. In addition, for convenience
of notation, we define $\lcp[1]=\lcp[n+1] = -1$. The {\em Burrows-Wheeler
transform} $\bwt[1,n]$ of $\txx{}$ is defined by
$$
\bwt[i] = \begin{cases}
\tx[n] & \mbox{if } \saxx[i]=1\\
\tx[\saxx[i]-1] & \mbox{if } \saxx[i]>1.
\end{cases}
$$
In other words, $\bwt[1,n]$ is the permutation of $\txx{}$ in which the
position of $\txx{}[j]$ coincides with the lexicographic rank of
$\txx{}[j+1,n]$ (or of $\txx{}[1,n]$ if $j=n$) in the suffix array. In
accordance with the literature we call such string the {\em context} of
$\txx{}[j]$. See Figure~\ref{fig:BWTetc} for an example.

\begin{figure}[t]
\def\xz{$\bullet$}
\def\xy{\$}
\setlength{\tabcolsep}{7pt}
\begin{center}\sf
\begin{tabular}[t]{ccc}
\begin{tabular}[t]{|r|c|l|}\hline
lcp&bwt& {\em context}\\\hline
 -1 & b & \xy       \\
  0 & c & ab\xy     \\
  2 &\xy& abcab\xy  \\
  0 & a & b\xy      \\
  1 & a & bcab\xy   \\
  0 & b & cab\xy    \\
 -1 &   & \\\hline
\end{tabular}&
\begin{tabular}[t]{|r|c|l|}\hline
lcp&bwt& {\em context}\\\hline
 -1 & c & \xz       \\
  0 &\xz& aabcabc\xz\\
  1 & c & abc\xz    \\
  3 & a & abcabc\xz \\
  0 & a & bc\xz     \\
  2 & a & bcabc\xz  \\
  0 & b & c\xz      \\
  1 & b & cabc\xz   \\
 -1 &   & \\\hline
\end{tabular}&
\begin{tabular}[t]{|r|r|c|l|}\hline
id &$\lcp_{01}$&$\mbwt_{01}$& {\em context}\\\hline
 0 & -1 & b & \xy\\
 1 &  0 & c & \xz       \\
 1 &  0 &\xz& aabcabc\xz\\
 0 &  1 & c & ab\xy\\
 1 &  2 & c & abc\xz    \\
 0 &  3 &\xy& abcab\xy\\
 1 &  5 & a & abcabc\xz \\
 0 &  0 & a & b\xy\\
 1 &  1 & a & bc\xz     \\
 0 &  2 & a & bcab\xy\\
 1 &  4 & a & bcabc\xz  \\
 1 &  0 & b & c\xz      \\
 0 &  1 & b & cab\xy\\
 1 &  3 & b & cabc\xz   \\
   & -1 &   & \\\hline
\end{tabular}
\end{tabular}
\end{center}
\caption{LCP array and BWT for $\txx{0}=\stri{abcab\xx}$
and $\txx{1}=\stri{aabcabc\xz}$,
and multi-string BWT and corresponding LCP array for the
same strings. Column {\sf id} shows, for each entry of
$\mbwt_{01} = \stri{bc\xz cc\xy aaaabbb}$ whether it comes
from $\txx{0}$ or $\txx{1}$.}\label{fig:BWTetc}
\end{figure}


The longest common prefix (LCP) array, and Burrows-Wheeler transform (BWT)
are fundamental components of a wide class of compressed full text indices.
We are interested in the generalization of these data structures when more
than one string is involved. Let $\tz[1,\nz]$ and $\tone[1,\none]$ denote a
pair of strings such that $\tz[\nz] = \eosz$ and $\tone[\none] = \eosone$
where $\eosz < \eosone$ are two symbols not appearing elsewhere in $\tz$ and
$\tone$ and smaller than any other symbol. One can consider the concatenation
$\txx{01}[1,\nz+\none] = \txx{0}\txx{1}$ and define LCP array and BWT for it.
However, for algorithmic reasons it is more convenient to consider the
following slightly different BWT definition, first introduced
in~\cite{tcs/MantaciRRS07} and later used for example
in~\cite{tcs/BauerCR13,jda/CoxGRS16,cpm/LouzaTC13,bioinformatics/HoltM14}.
Let $\SA_{01}[1,\nz+\none]$ denote the suffix array of the concatenation
$\tz\tone$. The {\em multi-string} BWT of $\tz$ and $\tone$, denoted by
$\mbwt_{01}[1,\nz+\none]$, is defined by
$$
\mbwt_{01}[i] =
\begin{cases}
\tz[\nz]                      & \mbox{if } \SA_{01}[i] = 1\\
\tz[\SA_{01}[i] - 1]       & \mbox{if } 1 < \SA_{01}[i] \leq \nz\\
\tone[\none]                    & \mbox{if } \SA_{01}[i] = \nz + 1\\
\tone[\SA_{01}[i]-\nz - 1] & \mbox{if } \nz +1 < \SA_{01}[i].
\end{cases}
$$
In other words, $\mbwt_{01}[i]$ is the symbol preceding the $i$-th
lexicographically larger suffix, with the exception that if $\SA_{01}[i] = 1$
then $\mbwt_{01}[i] = \eosz$  and if $\SA_{01}[i] = \nz+1$ then
$\mbwt_{01}[i] = \eosone$. Hence, $\mbwt_{01}[i]$ always comes from the
string ($\tz$ or $\tone$) that prefixes the $i$-th largest suffix (see again
Fig.~\ref{fig:BWTetc}).

The above notion of multi-string BWT can be immediately generalized to define
$\mbwt_{01\cdots k}$ for a family of distinct strings $\txx{0},
\txx{1},\ldots, \txx{k}$. Essentially $\mbwt_{01\cdots k}$ is a permutation
of the symbols in $\txx{0}, \txx{1},\ldots,\txx{k}$ such that the position in
$\mbwt_{01\cdots k}$ of $\txx{i}[j]$ is given by the lexicographic rank of
its context $\txx{i}[j+1,n_i]$ (or $\txx{i}[1,n_i]$ if $j=n_i$). Note that
the multi-string BWT $\mbwt_{01\cdots k}$ defined above is related to
$\mbwt(\txx{0}\txx{1}\cdots \txx{k})$, ie the single-string BWT of the
concatenation $\txx{0}\txx{1}\cdots \txx{k}$. Indeed, since they are both
defined in terms of the suffix array of $\txx{0}\txx{1}\cdots \txx{k}$, from
$\mbwt_{01\cdots k}$ we get $\mbwt(\txx{0}\txx{1}\cdots \txx{k})$ replacing
$\eosx{0}$ with $\eosx{k}$, $\eosx{1}$ with $\eosx{0}$, $\eosx{2}$ with
$\eosx{1}$ and so on.

Given the concatenation $\tz\tone$ and its suffix array
$\SA_{01}[1,\nz+\none]$, we consider the corresponding LCP array
$\lcp_{01}[1,\nz+\none+1]$ defined as in~\eqref{eq:lcpdef} (see again
Fig.~\ref{fig:BWTetc}). Note that, for $i=2,\ldots,\nz+\none$, $\lcp_{01}[i]$
gives the length of the longest common prefix between the contexts of
$\mbwt_{01}[i]$ and $\mbwt_{01}[i-1]$. This definition can be immediately
generalized to a family of $k$ strings to define the LCP array
$\lcp_{01\cdots k}$ associated to the multi-string BWT $\mbwt_{01\cdots k}$.

\section{The \hm\ algorithm revisited}

In~\cite{bioinformatics/HoltM14} Holt and McMillan introduced a simple and
elegant algorithm, we call it the \hm\ algorithm, to merge multi-string BWTs
as defined above.

Given $\mbwt_{01\cdots k}$ and $\mbwt_{k+1\,k+2\,\cdots h}$ the algorithm
computes $\mbwt_{01\cdots h}$. The computation does not explicitly need the
strings $\txx{0}, \txx{1}, \ldots, \txx{h}$ but only the BWTs to be merged.
For simplicity of notation we describe the algorithm assuming we are merging
two single string BWTs $\bwtz = \bwt(\tz)$ and $\bwto=\bwt(\tone)$; the
algorithm does not change in the general case where the input are multi
string BWTs. Note also that the algorithm can be easily adapted to merge more
than two BWTs at the same time, that is to compute directly $\mbwt_{01\cdots
k}$ given $\bwtx{0}, \bwtx{1}, \ldots, \bwtx{k}$.

\ignore{ Let $\tz$ and $\tone$ denote two strings of length respectively
$\nz=|\tz|$ and $\none = |\tone|$. The \hm\ algorithm computes $\mbwt_{01}$
as defined above given the single string BWTs $\bwtz = \bwt(\tz)$ and
$\bwto=\bwt(\tone)$.

in such a way that the position of $\txx{0}[j]$ (resp. $\txx{0}[j]$) is given
by the lexicographic rank of $\txx{1}[j+1,n_1]$ (resp. $\txx{1}[j+1,n_1]$). }

Computing $\mbwt_{01}$ amounts to sorting the symbols of $\bwtz$ and $\bwto$
according to the lexicographic order of their contexts, where the context of
symbol $\bwtz[i]$ (resp. $\bwto[i]$) is $\tz[\sa_0[i],n_0]$ (resp.
$\tone[\sa_1[i],n_1]$). By construction, the symbols in $\bwtz$ and $\bwto$
are already sorted by context, hence to compute $\mbwt_{01}$ we only need to
merge $\bwtz$ and $\bwto$ without changing the relative order of the symbols
within the two sequences.

The \hm\ algorithm works in successive phases. After the $h$-th phase the
entries of $\bwtz$ and $\bwto$ are sorted on the basis of the first $h$
symbols of their context. More formally, the output of the $h$-th phase is a
binary vector $\bv{h}$ containing $n_0=|\tz|$ \zerob's and $n_1 = |\tone|$
\oneb's and such that the following property holds.

\begin{property}\label{prop:hblock}
For $i=1,\ldots, n_0$, $j=1,\ldots n_1$ the $i$-th \zerob\ precedes the
$j$-th \oneb\ in $\bv{h}$ iff
\begin{equation}\label{eq:hblock}
\tz[\sa_0[i], \sa_0[i] + h -1] \;\preceq\; \tone[\sa_1[j], \sa_1[i] + h -1]
\end{equation}
(recall that according to our notation if $\sa_0[i] + h -1>n_0$ then
$\tz[\sa_0[i], \sa_0[i] + h -1]$ coincides $\tz[\sa_0[i], n_0]$, and
similarly for $\tone$).\qed
\end{property}

Following Property~\ref{prop:hblock} we identify the $i$-th \zerob\ in
$\bv{h}$ with $\bwtz[i]$ and the $j$-th \oneb\ in $\bv{h}$ with $\bwto[j]$ so
that to $\bv{h}$ corresponds a permutation of $\mbwt_{01}$.
Property~\ref{prop:hblock} is equivalent to state that we can logically
partition $\bv{h}$ into $\kh+1$ blocks
\begin{equation}\label{eq:Zblocks}
\bv{h}[1,\ell_1],\; \bv{h}[\ell_1+1, \ell_2],\; \ldots,\;
\bv{h}[\ell_\kh+1,n_0+n_1]
\end{equation}
such that each block corresponds to a set of $\mbwt_{01}$ symbols whose
contexts are prefixed by the same length-$h$ string (the symbols with a
context of length less than $h$ are contained in singleton blocks). Within
each block the symbols of $\bwtz$ precede those of $\bwto$, and the context
of any symbol in block $\bv{h}[\ell_j+1, \ell_{j+1}]$ is lexicographically
smaller than the context of any symbol in block $\bv{h}[\ell_k+1,
\ell_{k+1}]$ with $k>j$.

The \hm\ algorithm initially sets $\bv{0} = \zerox^{\nz} \onex^{\none}$:
since the context of every $\mbwt_{01}$ symbol is prefixed by the same
length-0 string (the empty string), there is a single block containing all
$\bwt_{01}$ symbols. At phase $h$ the algorithm computes $\bv{h+1}$ from
$\bv{h}$ using the procedure in Figure~\ref{fig:HMalgo}. For completeness we
report the proof of the correctness of the \hm\ algorithm, which is a
restatement of Lemma~3.2 in~\cite{bioinformatics/HoltM14} using our notation.

\algnewcommand\KwTo{\textbf{to }} \algnewcommand\KwAnd{\textbf{and }}
\begin{figure}
\hrule\smallbreak
\begin{algorithmic}[1]
\State Initialize array $F[1,|\A|]$
\State $k_0 \gets 1$; $k_1 \gets 1$  \Comment{Init counters for $\bwtz$ and $\bwto$}
\For{$k \gets 1$ \KwTo $n_0+ n_1$}
    \State $b \gets \bv{h-1}[k]$\Comment{Read bit $b$ from $\bv{h-1}$}
    \If {$b = 0$} \Comment{Get symbol from $\bwtz$ or $\bwto$ according to $b$}
      \State $c \gets \bwtz[k_0{\mathsf ++}]$\label{line:c0}
    \Else
      \State $c \gets \bwto[k_1{\mathsf ++}]$\label{line:c1}
    \EndIf
    \State $j \gets F[c]{\mathsf ++}$ \Comment{Get destination for $b$ according to symbol $c$}
    \State $\bv{h}[j] \gets b$ \Comment{Copy bit $b$ to $\bv{h}$}
\EndFor
\end{algorithmic}
\smallbreak\hrule
\caption{Main loop of algorithm \hm\ for computing $\bv{h}$ given
$\bv{h-1}$. Array $F$ is initialized so that $F[c]$ contains
the number of occurrences
of symbols smaller than $c$ in $\bwtz$ and $\bwto$ plus one. Hence,
the bits stored in $\bv{h}$ immediately after reading symbol $c$
are stored in positions from $F[c]$ to $F[c+1]-1$ of $\bv{h}$.}\label{fig:HMalgo}
\end{figure}

\begin{lemma}
For $h=0,1,2,\ldots$ the bit vector $\bv{h}$ satisfies
Property~\ref{prop:hblock}.
\end{lemma}

\begin{proof}
We prove the result by induction. For $h=0$, $\delta=0,1$
$\txx{\delta}[\sa_\delta[i],\sa_\delta[i]-1]$ is the empty string
so~\eqref{eq:hblock} is always true and Property~\ref{prop:hblock} is
satisfied by $\bv{0} = \zerox^\nz \onex^\none$.

To prove the ``if'' part, let $h>0$ and let $1 \leq v < w \leq \nz+\none$
denote two indexes such that $\bv{h}[v]$ is the $i$-th \zerob\ and
$\bv{h}[w]$ is the $j$-th \oneb\ in $\bv{h}$. We need to show that under
these assumptions inequality~\eqref{eq:hblock} on the lexicographic order
holds.

Assume first $\tz[\sa_0[i]]\neq \tone[\sa_1[j]]$. The hypothesis $v<w$
implies $\tz[\sa_0[i]]< \tone[\sa_1[j]]$ hence~\eqref{eq:hblock} certainly
holds.

Assume now $\tz[\sa_0[i]] = \tone[\sa_1[j]]$. We preliminary observe that it
must be $\sa_0[i]\neq\nz$ and $\sa_1[i]\neq\none$: otherwise we would have
$\tz[\sa_0[i]]=\xx_0$ or $\tone[\sa_1[j]]=\xx_1$ which is impossible since
these symbols appear only once in $\tz$ and $\tone$.

Let $v'$, $w'$ denote respectively the value of the main loop variable~$k$ in
the procedure of Figure~\ref{fig:HMalgo} when the entries $\bv{h}[v]$ and
$\bv{h}[w]$ are written (hence, during the scanning of $\bv{h-1}$). The
hypothesis $v<w$ implies ${v}' < {w}'$. By construction
$\bv{h-1}[{v}']=\zerox$ and $\bv{h-1}[{w}']=\onex$. Say ${v}'$ is the $i'$-th
\zerob\ in $\bv{h-1}$ and ${w}'$ is the $j'$-th \oneb\ in $\bv{h-1}$. By the
inductive hypothesis on $\bv{h-1}$ we have
\begin{equation}\label{eq:hblock2}
\tz[\sa_0[i'], \sa_0[i'] + h -2] \;\preceq\; \tone[\sa_1[j'], \sa_1[j'] + h -2],
\end{equation}
The fundamental observation is that, being $\sa_0[i]\neq\nz$ and
$\sa_1[i]\neq\none$, it is
$$
\sa_0[i'] = \sa_0[i] + 1\qquad\mbox{and}\qquad
\sa_1[j'] = \sa_1[j] + 1.
$$
Since
\begin{align}
\tz[\sa_0[i],\sa_0[i]+h-1]\; &= \;\tz[\sa_0[i]]\tz[\;\sa_0[i'],\sa_0[i']+h-2]\\
\tone[\sa_1[j],\sa_1[j]+h-1]\; &=\; \tone[\sa_1[j]]\;\tone[\sa_1[j'],\sa_1[j']+h-2]
\end{align}
combining $\tz[\sa_0[i]] = \tone[\sa_1[j]]$ with~\eqref{eq:hblock2} gives
us~\eqref{eq:hblock}.

For the ``only if'' part assume~\eqref{eq:hblock} holds. We need to prove
that in $\bv{h}$ the $i$-th \zerob\ precedes the $j$-th \oneb. If
$\tz[\sa_0[i]] < \tone[\sa_1[j]]$ the proof is immediate. If $\tz[\sa_0[i]] =
\tone[\sa_1[j]]$, we must have
$$
\tz[\sa_0[i]+1,\sa_0[i]+h-1] \preceq
\tone[\sa_1[j]+1,\sa_1[j]+h-1].
$$
By induction, if $\sa_0[i'] = \sa_0[i] + 1$ and $\sa_1[j']= \sa_1[j]+1$ in
$\bv{h-1}$ the $i'$-th \zerob\ precedes the $j'$-th \oneb. During phase~$h$,
the $i$-th \zerob\ in $\bv{h}$ is written when processing the $i'$-th \zerob\
of $\bv{h-1}$, and the $j$-th \oneb\ in $\bv{h}$ is written when processing
the $j'$-th \oneb\ of $\bv{h-1}$. Since in $\bv{h-1}$ the $i'$-th \zerob\
precedes the $j'$-th \oneb\ and
$$
\bwtz[i']=\tz[\sa_0[i]] = \tone[\sa_1[j]] = \bwto[j']
$$
in $\bv{h}$ their relative order does not change and the $i$-th \zerob\
precedes the $j$-th \oneb\ as claimed.
\end{proof}

\begin{figure}
\hrule\smallbreak
\begin{algorithmic}[1]
\State Initialize arrays $F[1,|\A|]$ and $\Bid[1,|\A|]$\label{line:init}
\State $k_0 \gets 1$; $k_1 \gets 1$  \Comment{Init counters for $\bwtz$ and $\bwto$}
\For{$k \gets 1$ \KwTo $n_0+ n_1$}
    \If{$B[k]\neq \sbot $ \KwAnd $B[k]\neq h$}\label{line:B=0h}
      \State $\bid\gets k$\Comment{A new block of $\bv{h-1}$ is starting}\label{line:blockstart}
    \EndIf
    \State $b \gets \bv{h-1}[k]$\Comment{Read bit $b$ from $\bv{h-1}$}\label{line:block_process_start}
    \If {$b = 0$} \Comment{Get symbol from $\bwtz$ or $\bwto$ according to $b$}
      \State $c \gets \bwtz[k_0{\mathsf ++}]$\label{line:c0}
    \Else
      \State $c \gets \bwto[k_1{\mathsf ++}]$\label{line:c1}
    \EndIf
    \State $j \gets F[c]{\mathsf ++}$ \Comment{Get destination for $b$ according to symbol $c$}
    \State $\bv{h}[j] \gets b$ \Comment{Copy bit $b$ to $\bv{h}$}
    \If{$\Bid[c]\neq \bid$}
      \State $\Bid[c]\gets \bid$\Comment{Update block id for symbol $c$}
      \If{$B[j] = \sbot$} 
        \State$B[j] = h$\Comment{A new block of $\bv{h}$ will start here}\label{line:writeh}
      \EndIf
    \EndIf \label{line:block_process_end}
\EndFor
\end{algorithmic}
\smallbreak\hrule
\caption{Main loop of the \hm\ algorithm modified for the computation of
the $\lcp$ values. At line~\ref{line:init}
for each symbol $c$ we set $\Bid[c] = -1$ and $F[c]$ as in
Figure~\ref{fig:HMalgo}. At the beginning of the algorithm we initialize the
array $B[0,\nz+\none]$ as $B = 1\:0^{\nz+\none-1}\:1$.}\label{fig:HMlcp}
\end{figure}

We now show that with a simple modification to the \hm\ algorithm it is
possible to compute, in addition to $\bwtzo$, also the LCP array $\lcpzo$
defined in Section~\ref{sec:notation}. Our strategy for computing LCP values
consists in keeping explicit track of the logical blocks we have defined for
$\bv{h}$ and represented in~\eqref{eq:Zblocks}. More precisely, we maintain
an integer array $B[1,n_0+n_1+1]$ such that at the end of phase $h$ it is
$B[i]\neq 0$ iff a block of $\bv{h}$ starts at position~$i$. The use of such
integer array is shown in Figure~\ref{fig:HMlcp}. Note that: $(i)$ initially
we set $B = 1\:0^{\nz+\none-1}\:1$ and once an entry in $B$ becomes nonzero
it is never changed,  $(ii)$ during phase $h$ we only write to $B$ the value
$h$, $(iii)$ in the test at Line~\ref{line:B=0h} the value $h$ is equivalent
to 0, hence the values written during phase $h$ influence the algorithm only
in subsequent phases. The following lemma shows that the nonzero values of
$B$ at the end of phase~$h$ mark the boundaries of $\bv{h}$'s logical blocks.

\begin{lemma} \label{lemma:B}

For any $h\geq 0$, let $\ell$, $m$ be such that $1 \leq \ell \leq m \leq
\nz+\none$ and
\begin{equation}\label{eq:lcpblock}
\lcpzo[\ell] < h,\quad \min(\lcpzo[\ell+1], \ldots, \lcpzo[m]) \geq h,
\quad \lcpzo[m+1] < h.
\end{equation}
Then, at the end of phase $h$ the array $B$ is such that
\begin{equation}\label{eq:Bblock}
B[\ell]\neq 0, \quad B[\ell+1] = \cdots = B[m] = 0,
\quad B[m+1] \neq 0
\end{equation}
and $\bv{h}[\ell,m]$ is one of the blocks in~\eqref{eq:Zblocks}.
\end{lemma}

\begin{proof}
We prove the result by induction on $h$. For $h=0$, hence before the
execution of the first phase, \eqref{eq:lcpblock} is only valid for $\ell=1$
and $m=\nz+\none$ (recall we defined $\lcpzo[1]=\lcpzo[\nz+\none+1]=-1$).
Since initially $B=1\: 0^{\nz+\none-1}\:1$ our claim holds.

Suppose now that~\eqref{eq:lcpblock} holds for some $h>0$. Let
$s=\tzo[\sazo[\ell],\sazo[\ell] + h-1]$; by~\eqref{eq:lcpblock} $s$ is a
common prefix of the suffixes starting at positions $\sazo[\ell]$,
$\sazo[\ell+1]$, \ldots, $\sazo[m]$, and no other suffix of $\tzo$ is
prefixed by~$s$. By Property~\ref{prop:hblock} the \zerob's and \oneb's in
$\bv{h}[\ell,m]$ corresponds to the same set of suffixes That is, if $\ell
\leq v \leq m$ and $\bv{h}[v]$ is the $i$th \zerob\ (resp. $j$th \oneb) of
$\bv{h}$ then the suffix starting at $\tz[\sa_0[i]]$ (resp. $\to[\sa_1[j]]$)
is prefixed by $s$.

To prove~\eqref{eq:Bblock} we start by showing that, if $\ell < m$, then at
the end of phase $h-1$ it is $B[\ell+1] = \cdots = B[m] = 0$. To see this
observe that the range $\sazo[\ell,m]$ is part of a (possibly) larger range
$\sazo[\ell',m']$ containing all suffixes prefixed by the length $h-1$ prefix
of $s$. By inductive hypothesis, at the end of phase $h-1$ it is $B[\ell'+1]
= \cdots = B[m'] = 0$ which proves our claim since $\ell'\leq \ell$ and $m
\leq m'$.

To complete the proof, we need to show that during phase $h$: $(i)$ we do not
write a nonzero value in $B[\ell+1,m]$ and $(ii)$ we write a nonzero to
$B[\ell]$ and $B[m+1]$ if they do not already contain a nonzero. Let $c=s[0]$
and $s'=s[1,h-1]$ so that $s = cs'$. Consider now the range $\sazo[e,f]$
containing the suffixes prefixed by $s'$. By inductive hypothesis at the end
of phase $h-1$ it is
\begin{equation}\label{eq:Bblock_inproof}
B[e]\neq 0, \quad B[e+1] = \cdots = B[f] = 0,
\quad B[f+1] \neq 0.
\end{equation}
During iteration $h$, the bits in $\bv{h}[\ell,m]$ are possibly changed only
when we are scanning the region $\bv{h-1}[e,f]$ and we find an entry
$b=\bv{h-1}[k]$, $e\leq k \leq f$, such that the corresponding value in
$\bwtx{b}$ is $c$. Note that by~\eqref{eq:Bblock_inproof} as soon as $k$
reaches $e$ the variable $\bid$ changes and becomes different from all values
stored in $\Bid$. Hence, at the first occurrence of symbol $c$ the value $h$
will be stored in $B[\ell]$ (Line~\ref{line:writeh}) unless a nonzero is
already there. Again, because of~\eqref{eq:Bblock_inproof}, during the
scanning of $\bv{h-1}[e,f]$ the variable $\bid$ does not change so subsequent
occurrences of $c$ will not cause a nonzero value to be written to
$B[\ell+1,m]$. Finally, as soon as we leave region $\bv{h-1}[e,f]$ and $k$
reaches $f+1$, the variable $\bid$ changes again and at the next occurrence
of $c$ a nonzero value will be stored in $B[m+1]$. If there are no more
occurrences of $c$ after we leave region $\bv{h-1}[e,f]$ then either
$\sazo[m+1]$ is the first suffix array entry prefixed by symbol $c+1$ or
$m+1=\nz+\none+1$. In the former case $B[m+1]$ gets a nonzero value at phase
1, in the latter case $B[m+1]$ gets a nonzero when we initialize array $B$.

This completes the proof
\end{proof}

\begin{corollary}\label{cor:lcp}

For $i=2,\ldots,\nz+\none$, if $\lcpzo[i] = \ell$, then starting from the end
of phase $\ell+1$ it is $B[i]=\ell+1$.

\end{corollary}

\begin{proof}
By Lemma~\ref{lemma:B} we know that $B[i]$ becomes nonzero only after phase
$\ell+1$. Since at the end of phase $\ell$ it is still $B[i]=0$ during phase
$\ell+1$ $B[i]$ gets the value $\ell+1$ which is never changed in successive
phases.
\end{proof}

The above corollary suggests the following algorithm to compute $\bwtzo$ and
$\lcpzo$: repeat the procedure of Figure~\ref{fig:HMlcp} until the phase $h$
in which all entries in $B$ become nonzero. At that point $\bv{h}$ describes
how $\bwtz$ and $\bwto$ should be merged to get $\bwtzo$ and for $i=2,\ldots,
\nz+\none$ $\lcpzo[i] = B[i] - 1$. The above strategy requires a number of
iterations, each one taking $\Oh(\nz+\none)$ time, equal to the maximum of
the $\lcp$ values for an overall complexity of
$\Oh((\nz+\none)\mathsf{maxlcp}_{01})$, where $\mathsf{maxlcp}_{01}=\max_i
\lcpzo[i]$. In the next section we describe a much faster algorithm that goes
beyond this simple strategy and avoids to re-process the portions of $B$ and
$\bv{h}$ which are no longer relevant for the computation of the final
result.

\section{The \gap\ algorithm}

\begin{definition}
If $B[\ell]\neq 0$, $B[m+1]\neq 0$ and $B[\ell+1] = \cdots = B[m] = 0$, we
say that block $\bv{h}[\ell,m]$ is {\em monochrome} if it contains only
\zerob's or only \oneb's.\qed
\end{definition}

Since a monochrome block only contains suffixes from either $\tz$ or $\tone$,
whose relative order and LCP's are known, it does not need to be further
modified. This intuition is formalized by the following lemmas.

\begin{lemma}\label{lemma:monochrome}
If at the end of phase $h$ bit vector $\bv{h}$ contains only monochrome
blocks we can compute $\bwtzo$ and $\lcpzo$ in $\Oh(\nz + \none)$ time.
\end{lemma}

\begin{proof}

By Property~\ref{prop:hblock}, if we identify the $i$-th \zerob\ in $\bv{h}$
with $\bwtz[i]$ and the $j$-th \oneb\ with $\bwto[j]$ the only elements which
could be out of order (ie not correctly sorted by context) are those within
the same block. However, if the blocks are monochrome all elements belongs to
either $\bwtz$ or $\bwto$ so their relative order is correct.

To compute $\lcpzo$ we observe that if $B[i]\neq 0$ then by (the proof of)
Corollary~\ref{cor:lcp} it is $\lcpzo[i] = B[i] -1$. If instead $B[i]=0$ we
are inside a block hence $\sazo[i-1]$ and $\sazo[i-1]$ belongs to the same
string $\tz$ or $\tone$ and their $\lcp$ is directly available in $\lcpz$ or
$\lcpo$.
\end{proof}

\begin{lemma}\label{lemma:skip}

Suppose that, at the end of phase $h$, $\bv{h}[\ell,m]$ is a monochrome
block. Then $(i)$ for $g>h$, $\bv{g}[\ell,m] = \bv{h}[\ell,m]$, and $(ii)$
processing $\bv{h}[\ell,m]$ during phase $h+1$ creates a set of monochrome
blocks in $\bv{h+1}$.

\end{lemma}


\begin{proof}

The first part of the Lemma follows from the observation that subsequent
phases of the algorithm will only reorder the values within a block (and
possibly create new sub-blocks); but if a block is monochrome the reordering
will not change its actual content.

For the second part, we observe that during phase $h+1$ as $k$ goes from
$\ell$ to $m$ the algorithm writes to $\bv{h+1}$ the same value which is in
$\bv{h}[\ell,m]$. Hence, a new monochrome block will be created for each
distinct symbol encountered (in $\bwtz$ or $\bwto$) as $k$ goes through the
range $[\ell,m]$.
\end{proof}

The lemma implies that, if block $\bv{h}[\ell,m]$ is monochrome at the end of
phase $h$, starting from phase $g=h+2$ processing the range $[\ell,m]$ will
not change $\bv{g}$ with respect to $\bv{g-1}$. Indeed, by the lemma the
monochrome blocks created in phase $h+1$  do not change in subsequent phases
(in a subsequent phase a monochrome block can be split in sub-blocks, but the
actual content of the bit vector does not change). The above observation
suggests that, after we have processed block $\bv{h+1}[\ell,m]$ in phase
$h+1$, we can mark it as {\em \useless} and avoid to process it again. As the
computation goes on, more and more blocks become \useless. Hence, in the
generic phase $h$ instead of processing the whole $\bv{h-1}$ we process only
the blocks which are still ``active'' and skip \useless\ blocks. Adjacent
\useless\ blocks are merged so that among two active blocks there is at most
one \useless\ block (the {\em gap} that gives the name to the algorithm). The
overall structure of a single phase is shown in Figure~\ref{fig:gap2}. The
algorithm terminates when there are no more active blocks since this implies
that all blocks have become monochrome and by Lemma~\ref{lemma:monochrome} we
are able to compute $\bwtzo$ and $\lcpzo$.

\begin{figure}
\hrule\smallbreak
\begin{algorithmic}[1]
\If{(next block is \useless)}
    \State{skip it}\label{line:skip}
\Else
    \State{process block}\label{line:process_block}
    \If{(processed block is monochrome)}
      \State mark it \useless
    \EndIf
\EndIf
\If{(last two blocks are \useless)}
    \State{merge them}
\EndIf
\end{algorithmic}
\smallbreak\hrule
\caption{Main loop of the \gap\ algorithm. The processing of active
blocks at Line~\ref{line:process_block} is done as in
Lines~\ref{line:block_process_start}--\ref{line:block_process_end}
of Figure~\ref{fig:HMlcp}.}\label{fig:gap2}
\end{figure}

We point out that at Line~\ref{line:skip} of the \gap\ algorithm we cannot
simply skip an \useless\ block ignoring its content. To keep the algorithm
consistent we must correctly update the global variables of the main loop,
i.e. the array $F$ and the pointers $k_0$ and $k_1$ in
Figure~\ref{fig:HMlcp}. To this end a simple approach is to store for each
(merged) \useless\ block the number of occurrences $o_c$ of each symbol
$c\in\A$ in it and the pair $(r_0,r_1)$ providing the number of \zerob's and
\oneb's in the block. When the algorithm reaches an \useless\ block, $F$,
$k_0$, $k_1$ are updated setting $k_0 \gets k_0 + r_0$, $k_1 \gets k_1 + r_1$
and $\forall c$ $F[c] \gets F[c] + o_c$.

The above scheme for handling \useless\ blocks is simple and probably
effective in most cases. However, using $\Oh(|\A|)$ time to skip an
\useless\ block is not competitive in terms of worst case complexity. A
better alternative is to build a wavelet tree for $\bwtz$ and $\bwto$ at the
beginning of the algorithm. Then, for each \useless\ block we store only the
the pair $(r_0,r_1)$. When we reach an \useless\ block we use such pair to
update $k_0$ and $k_1$. The array $F$ is not immediately updated: Instead we
maintain two global arrays $L_0[1,|\A|]$ and $L_1[1,|\A|]$ such that $L_0[c]$
and $L_1[c]$ store the value of $k_0$ and $k_1$ at the time the value $F[c]$
was last updated. At the {\em first} occurrence of a symbol $c$ inside an
active block we update $F[c]$ adding to it the number of occurrences of $c$
in $\bwtz[L_o[c]+1,k_0]$ and $\bwto[L_1[c]+1,k_1]$ that we compute in
$\Oh(\log|\A|)$ time using the wavelet trees. Using this lazy update
mechanism, handling \useless\ blocks adds a $\Oh(\min(\ell,|\A|) \log|\A|)$
additive slowdown to the cost of processing an active block of length~$\ell$.

\begin{theorem}
Given $\bwtz, \lcpz$ and $\bwto, \lcpo$ the \gap\ algorithm computes $\bwtzo$
and $\lcpzo$ in $\Oh(  \log(|\A|) (\nz + \none) \avelcp )$ time, where
$\avelcp=(\sum_i \lcpzo[i])/(\nz+\none)$ is the average LCP of the string
$\tzo$.

\end{theorem}

\begin{proof}
The correctness follows from the above discussion. For the analysis of the
running time we reason as in~\cite{bcb/HoltM14} and observe that the sum,
over all phases, of the length of all active blocks is bounded by $\Oh(\sum_i
\lcpzo[i]) = \Oh((\nz+\none)\avelcp)$. In any phase, using the lazy update
mechanism,  the cost of processing an active block of length $\ell$ is
bounded by $\Oh(\ell \log(|\A|)$ and final time bound follows.
\end{proof}

We point out that our \gap\ algorithm is related to the version of the \hm\
algorithm described in~\cite[Sect.~2.1]{bcb/HoltM14}: Indeed, the sorting
operations are essentially the same in the two algorithms. The main
difference is that \gap\ keeps explicit track of the \useless\ blocks while
\hm\ keeps explicit track of the active blocks (called buckets
in~\cite{bcb/HoltM14}): this difference makes the non-sorting operations
completely different. An advantage of working with \useless\ blocks is that
they can be easily merged, while this is not the case for the active blocks
in \hm. Of course, the main difference is that \gap\ computes simultaneously
$\bwtzo$ and $\lcpzo$ while \hm\ only computes $\bwtzo$.

\end{document}